\documentclass[10pt]{IEEEtran}
\pdfoutput=1 
\usepackage{amsmath}
\usepackage{amsthm}
\usepackage{color}
\usepackage{graphicx}
\usepackage{calc}
\usepackage{import}
\usepackage{url}
\usepackage{wrapfig}
\usepackage[all]{xy}
\usepackage{enumerate}

\newcommand{\outputs}{\mathsf{outputs}}
\newcommand{\net}{\mathsf{net}}
\newcommand{\word}{\mathsf{word}}
\newcommand{\sent}[1]{\ensuremath{\mathtt{#1}}} %
\newcommand{\GG}{{\cal G}}
\newcommand\tuple[1]{\langle #1 \rangle}
\newcommand{\sset}[2]{\left\{~#1  \left|
      \begin{array}{l}#2\end{array}
    \right.     \right\}}

\newtheorem{lemma}{Lemma}
\newtheorem{definition}{Definition}
\newtheorem{theorem}{Theorem}
\newtheorem{conjecture}{Conjecture}
\newtheorem{proposition}{Proposition}
\newtheorem{example}{Example}
\newtheorem{corollary}{Corollary}

\begin{document}

\title{The Quest for Optimal Sorting Networks: Efficient Generation of
       Two-Layer Prefixes\thanks{Supported by 
       the Israel Science Foundation, grant 182/13 and by
       the Danish Council for Independent Research, Natural Sciences.}}
\author{\IEEEauthorblockN{Michael Codish\\}
\IEEEauthorblockA{Department of Computer Science\\
    Ben-Gurion University of the Negev\\
    PoB 653\\
    Beer-Sheva, Israel 84105\\[2ex]}
\and
\IEEEauthorblockN{Lu\'{\i}s Cruz-Filipe and Peter Schneider-Kamp\\}
\IEEEauthorblockA{Department of Mathematics and Computer Science\\
University of Southern Denmark\\
Campusvej 55\\
5230 Odense M, Denmark}}
\maketitle

\begin{abstract}
  Previous work identifying depth-optimal $n$-channel sorting networks
  for $9\leq n \leq 16$ is based on exploiting symmetries of the first
  two layers.  However, the naive generate-and-test approach typically
  applied does not scale.
  This paper revisits the problem of generating two-layer prefixes
  modulo symmetries. An improved notion of symmetry is provided and a
  novel technique based on regular languages and graph isomorphism is
  shown to generate the set of non-symmetric representations.  An
  empirical evaluation demonstrates that the new method outperforms
  the generate-and-test approach by orders of
  magnitude and easily scales until $n=40$.
\end{abstract}

\section{Introduction}

Sorting networks are a Computer Science classic. Based on a very
simple model, their underlying theory is surprisingly deep and
complex.  The study of sorting networks has intrigued computer
scientists since the middle 1950s. 
Informally, a sorting network is a comparator network that sorts all of
its inputs. A comparator network is a network constructed from $n$
channels that carry $n$ input values from ``left to right'' through a
sequence of comparators. A comparator is a component attached to a
pair of channels such that the pair of values coming in from the left
come out sorted on the right. Consecutive comparators can be viewed as
a ``parallel layer'' if no two touch on the same channel.
For an overview on sorting networks see for example, Knuth~\cite{Knuth73}
or Parberry~\cite{Parberry87}.

Ever since sorting networks were introduced, there has been a quest to
find optimal sorting networks for particular small numbers of inputs:
optimal depth networks (in the number of parallel layers), as well as
optimal size (in the number of comparators). In this paper we focus on
optimal depth sorting networks.

Even today, very little progress has been seen.
Optimal depth sorting networks for $n\leq 8$ are given by Knuth (1973),
Section~5.3.4 of~\cite{Knuth73}, which also details specific sorting
networks for $n\leq 16$ with the smallest depths known at the time. In
1991, Parberry~\cite{DBLP:journals/mst/Parberry91} showed that the
networks given by Knuth are optimal for $n = 9$ and $n = 10$.
Parberry's result was obtained by implementing an exhaustive
search with pruning based on symmetries in the first two layers of the
comparator networks, and executing the algorithm on a supercomputer
(consuming 200 hours of low priority computation).

In 2011, Morgenstern and Schneider~\cite{DBLP:conf/mbmv/MorgensternS11} applied
SAT solvers to search for optimal depth sorting networks, and were
able to reproduce the known results for $n<10$ with an acceptable
runtime, but still required 21 days of computation for
$n=10$, shredding any hope to achieve reasonable runtimes for $n \geq 11$.
Optimality for the cases $11\leq n\leq 16$ is shown by Bundala and
Z{\'a}vodn{\'y} (2014) in~\cite{DBLP:conf/lata/BundalaZ14}, first by
showing that $n=11$ requires at least depth $8$, and then by showing
that $n=13$ requires at least depth $9$. Their results are obtained
using a SAT solver, and are also based on identifying symmetries in
the first two layers of the sorting networks.

Both Parberry~\cite{DBLP:journals/mst/Parberry91} and then Bundala and
Z{\'a}vodn{\'y}~\cite{DBLP:conf/lata/BundalaZ14} consider the
following question: what is the smallest set $S$ of two-layer network
prefixes that need be considered in the search for minimal depth
sorting networks? In particular, such that, if no element of $S$ can be
extended to a sorting network of depth~$d$, then no depth $d$ sorting
network exists.
The approach
in~\cite{DBLP:conf/lata/BundalaZ14}
\linebreak
identified~$212$ two-layer network prefixes for $n=13$; however, the
calculation of this set required $32$ minutes of computation, and this
approach does not scale for larger values of~$n$.

In this paper, we show how to generate the same set of $212$ two-layer
prefixes for $n=13$ in ``under a second'' and,
following ideas presented in~\cite{DBLP:conf/lata/BundalaZ14},
improve results such that
only $117$ relevant two-layer prefixes need to be considered. Our approach also scales
well, i.e.\ we can compute the set of $34{,}486$ relevant prefixes
for $n=30$ in ``under a minute'', and that of relevant prefixes for
$n=40$ in around two hours.  Our main contribution here is to
illustrate how focusing on concepts of regular languages, graph
isomorphism, and symmetry breaking facilitates the efficient generation
of all two-layer prefixes modulo isomorphism of the networks.

\section{Preliminaries on Sorting Networks}
\label{sec:prelim}

A \emph{comparator network} $C$ with $n$ channels and depth $d$ is a
sequence $C = L_1;\ldots;L_d$ where each \emph{layer} $L_k$ is
a set of comparators $(i,j)$ for pairs of channels $i < j$. At each
layer, every channel may occur in at most one comparator.  A layer
is \emph{maximal} if it contains $\left\lfloor\frac n2\right\rfloor$
comparators.  The \emph{depth} of $C$ is the number of layers $d$, and
the \emph{size} of $C$ is the total number of comparators in its
layers. If $C_1$ and $C_2$ are comparator networks, then $C_1;C_2$
denotes the comparator network obtained by concatenating the layers of
$C_1$ and $C_2$; if $C_1$ has $m$ layers, it is an \emph{$m$-layer
  prefix} of $C_1;C_2$.
An input $\bar x\in\{0,1\}^n$ propagates through $C$ as follows:
$\bar x_0 = \bar x$, and for $0<k\leq d$, $\bar x_k$ is a permutation
of $\bar x_{k-1}$ obtained such that for each comparator $(i,j)\in
L_k$, the values at positions $i$ and $j$ of $\bar x_{k-1}$ are
reordered in $\bar x_k$ so that the value at position $i$ is not
larger than the value at position $j$. The output of the network for
input $\bar x$ is $C(\bar x)=\bar x_d$, and
$\outputs(C)=\sset{C(\bar x)}{\bar x\in\{0,1\}^n}$.  The
comparator network $C$ is a \emph{sorting network} if all elements of
$\outputs(C)$ are sorted (in ascending order).
The zero-one principle~(e.g.~\cite{Knuth73}) implies that a sorting
network also sorts any other totally ordered set, e.g.~integers.
The \emph{optimal sorting network problem} is about finding the smallest
depth and the smallest size of a sorting network for a given number of
channels $n$.

A \emph{generalized comparator network} is defined like a comparator network,
except that it may contain comparators $(i,j)$ with $i>j$, which
order their outputs in descending order, instead of ascending.
It is well known~(Exercise~5.3.4.16 in~\cite{Knuth73}) that generalized sorting
networks are no more powerful than sorting networks: a generalized
sorting network can always be \emph{untangled} into a (standard) sorting
network with the same size and depth.

\begin{wrapfigure}[11]{r}{0.15\textwidth}
  \vspace*{-2ex}
  $(a)$~\raisebox{-\height/2}{\makebox{\includegraphics{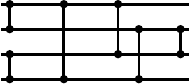}}}

  \vspace{1em}

  $(b)$~\raisebox{-\height/2}{\makebox{\includegraphics{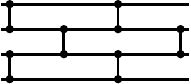}}}

  \vspace{1em}

  $(c)$~\raisebox{-\height/2}{\makebox{\includegraphics{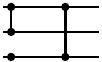}}}

  \vspace{1em}

  $(d)$~\raisebox{-\height/2}{\makebox{\includegraphics{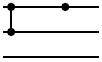}}}
  \label{ex:sn}
\end{wrapfigure}

Images~$(a)$ and~$(b)$ on the right depict sorting networks on four
channels, each consisting of four layers. The channels are indicated
as horizontal lines (with channel $4$ at the bottom), comparators are
indicated as vertical lines connecting a pair of channels, and input
values are assumed to propagate from left to right. Images~$(c)$ and~$(d)$
specify \emph{patterns}.
A pattern $P$ is a partially specified network: it is a set of
channels with comparators, but it may also include \emph{external comparators}.
These are singleton nodes representing a comparator connected to a channel not in $P$.  
A comparator network $C$ contains a pattern $P$ of depth~$d$ on $m$
channels if there are a depth $d$ prefix $C_1$ of $C$ and a subset
$\{c_1,\ldots,c_m\}$ of channels of $C_1$ such that: (i)~$c_i<c_j$ if
$i<j$; (ii)~if $P$ contains a comparator between channels $i$ and $j$
at layer~$1\leq k\leq d$, then $C_1$ contains a comparator between
channels $c_i$ and $c_j$ at layer~$k$; (iii)~if $P$ contains an
external comparator touching channel $i$ at layer~$1\leq k\leq d$,
then $C_1$ contains a comparator between channel $c_i$ and a channel
$c\not\in\{c_1,\ldots,c_m\}$ at layer~$k$; (iv)~$C_1$ 
contains no other comparators connecting to or between channels
$c_1,\ldots,c_m$.
The depth~$2$, three-channel pattern depicted in~$(c)$ occurs in network~$(a)$
but not in~$(b)$, while the pattern in~$(d)$ does not occur in
either network~$(a)$ or~$(b)$: its third channel is never used, while
all channels of~$(a)$ and~$(b)$ are used in the
first two layers.

We can use permutations $\pi$ on channels to manipulate (layers of)
comparator networks.  %
For a layer $L$, $\pi(L)$ contains the comparator $(\pi(i),\pi(j))$
iff $L$ contains $(i,j)$. If it is always the case that $\pi(i) <
\pi(j)$, then $\pi(L)$ is also a layer, otherwise it is a
\emph{generalized layer}. The extension to networks is
straightforward, and we write $C_1 \approx C_2$ ($C_1$ is equivalent
to $C_2$) iff there is a permutation $\pi$ such that $C_1$ is obtained
by untangling the (generalized) comparator network $\pi(C_2)$.  The
two networks~$(a)$ and~$(b)$ above are equivalent via the permutation
$(1\,3)(2\,4)$ and the application of the construction for untangling
described in~\cite{Knuth73}.

Parberry~\cite{DBLP:journals/mst/Parberry91} shows that the first
layer of a depth-optimal sorting network on $n$ channels can be
taken to consist of the comparators $(2k-1,2k)$ for $1 \leq k \leq
\left\lfloor \frac{n}{2} \right\rfloor$. We denote this layer by $F_n$.  The
networks~$(a)$ and~$(b)$ have first layer $F_4$.
In general, when $L_1;C$ is an $n$ channel
comparator network, we call a channel of $C$ ``min'' (``max'') if it
is connected to the minimum (maximum) output of a comparator in $L_1$,
and ``free'' if it does not occur in a comparator of~$L_1$.

We make use of the following two lemmata, which are proved
in~\cite{DBLP:conf/lata/BundalaZ14}.
The first lemma originates from~\cite{DBLP:journals/mst/Parberry91}.

\begin{lemma}%
  Let $\pi$ be a permutation
  such that $\pi(F_n) = F_n$ and let $L$ be a layer on $n$ channels
  such that $\pi(L)$ is a layer. If there is an $n$-channel
  sorting network of the form $F_n;L;C$ with depth $d$, then there is one of the form
  $F_n;\pi(L);C'$ with depth $d$.
\end{lemma}

\begin{lemma}
  \label{lem:outputs}
  Let $L_a$ and $L_b$ be layers on $n$ channels such that
  $\outputs(F_n;L_b) \subseteq \pi(\outputs(F_n;L_a))$ for
  some permutation~$\pi$. If there is a sorting network $F_n;L_a;C$ of
  depth $d$, then there is also a sorting network $F_n;L_b;C'$ of
  depth $d$.
\end{lemma}

\section{Saturation}

Bundala and Z{\'a}vodn{\'y}
introduce, in~\cite{DBLP:conf/lata/BundalaZ14}, the notion of a
saturated layer.
We call a comparator network saturated if its last layer is saturated.
The motivation is that, usually, adding a comparator to a network
decreases the set of its possible outputs, but not always. When the
network is saturated, adding a comparator to its last layer does not
decrease the set of its possible outputs.
This means that, when seeking a sufficient set of two layer networks
with which to search for depth-optimal sorting networks, one can
consider only saturated ones.
The definition of saturation in~\cite{DBLP:conf/lata/BundalaZ14} is
syntactic. In this section, we propose a semantic characterization,
and prove a syntactic criterion which is stronger than the one
proposed therein.
This means that we need to consider fewer two-layer networks.

\begin{definition}
  A comparator network $C$ is \emph{redundant} if there exists a
  network $C'$ obtained from $C$ by removing a comparator such that
  $\outputs(C')=\outputs(C)$.  A network $C$ is \emph{saturated} if it
  is non-redundant and every network $C'$ obtained by adding a
  comparator to the last layer of $C$ satisfies $\outputs(C')\not\subseteq\outputs(C)$.
\end{definition}

Parberry~\cite{DBLP:journals/mst/Parberry91} shows that the first
layer of a minimal-depth sorting network on $n$ channels can always be assumed to contain
$\left\lfloor\frac n2\right\rfloor$ comparators.  Also,
any comparator network that contains the same comparator at consecutive
layers is redundant.

\begin{theorem}
  \label{thm:sat-char}
  Let $C$ be a saturated two-layer network.  Then $C$ contains none of the
  following two-layer patterns.\medskip

  \noindent\hfill$\phantom{c}(1)$~\raisebox{-\height/2}{\makebox{\includegraphics{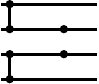}}}
  $(2a)$~\raisebox{-\height/2}{\makebox{\includegraphics{saturated-2a.pdf}}}
  $(2b)$~\raisebox{-\height/2}{\makebox{\includegraphics{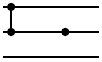}}}\hspace*\fill\\[\smallskipamount]
  
  \noindent\hfill$(2c)$~\raisebox{-\height/2}{\makebox{\includegraphics{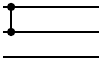}}}
  $(3a)$~\raisebox{-\height/2}{\makebox{\includegraphics{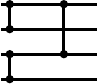}}}
  $(3b)$~\raisebox{-\height/2}{\makebox{\includegraphics{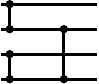}}}\hspace*\fill\\
\end{theorem}

\begin{proof}
  Although this formulation is more general, the proof of case~$(1)$ is the
  same as the first case of the proof of Lemma~8 of~\cite{DBLP:conf/lata/BundalaZ14}, and
  the proof of cases~$(2a)$, $(2b)$ and~$(2c)$ is the same as the second case of the same
  proof.

  For case~$(3a)$, assume that $C$ includes the given pattern and let the channels
  corresponding to those in the pattern be $a$, $b$, $c$ and~$d$.  Add a
  comparator between channels $b$ and $d$ to obtain a network $C'$ that
  includes the following pattern.\medskip

  \hfill\includegraphics{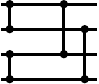}\hspace*\fill\medskip
  
  For a given input $i$ of $C'$, let $i_k$, $m_k$ and $o_k$ denote the values on
  channel $k$ respectively at input, after layer~$1$, and at output.
  Exchanging the values of $i_a$ with $i_c$ and $i_b$ with $i_d$ has the effect
  of exchanging $m_a$ with $m_c$ and $m_b$ with $m_d$.
  Then the output $o$ can be obtained as an output of $C$:
  \begin{itemize}
  \item if $m_b\leq m_d$, then $o$ is the output of $C$ corresponding to input
    $i$;
  \item if $m_d<m_b$, then $o$ is the output of $C$ corresponding to input $i'$
    obtained from $i$ by permuting $i_a$ with $i_c$, $i_b$ with $i_d$, and
    maintaining the value on all other channels.
  \end{itemize}
  Therefore $C$ is not saturated.

  For case~$(3b)$ the construction is the same, and the thesis follows by
  comparing $m_a$ with $m_c$.
\end{proof}

As it turns out, these are actually \emph{all} of the patterns that make
a comparator network with first layer $F_n$ non-saturated. We formalize
this observation in the following theorem.
\begin{theorem}
  \label{thm:sat-thm}
  If $C$ is a non-redundant two-layer network on $n$ channels with first
  layer $F_n$ containing none of the
  patterns in Theorem~\ref{thm:sat-char}, then $C$ is saturated.
\end{theorem}

\begin{proof}
  Let $C$ be a non-redundant two-layer comparator network, and assume that the second layer of $C$
  has at least two unused channels (otherwise there is nothing to prove).  If one
  of these channels were unused at layer~$1$, then the network would contain the
  pattern~$(2a)$, $(2b)$ or~$(2c)$. Thus, by Theorem~\ref{thm:sat-char}
  necessarily the two channels are connected at layer~$1$. Again from the same theorem,
  we know that they must be both min-channels or both max-channels (otherwise case~$(1)$ applies)
  and the channels they are connected to at layer~$1$ cannot be connected at
  layer~$2$, otherwise the network would be redundant.

  There are eight different cases to consider.  We
  detail the cases where the two unused channels are max channels. Assume that the
  four relevant channels are adjacent. This does not lose generality, since a first-layer preserving permutation
  can always be applied to $C$ to make this hold. Label the channels $a$, $b$, $c$
  and $d$ from top to bottom (so $(a,b)$ and $(c,d)$ are comparators at
  layer~$1$ and channels~$b$ and~$d$ are unused at layer~$2$).
  Let $k$ be the number of channels above~$a$ and $m$ be the number of channels below~$d$.
  Adding a comparator to $C$ yields $C'$ where $(b,d)$ is a comparator at
  layer~$2$.
  The four possibilities depend on whether channels~$a$ and~$c$ are min- or
  max-channels,
  and are represented in Figure~\ref{fig:sat-thm}.

  \begin{figure}[t]
    \hfill
    (i)~\raisebox{-\height/2}{\makebox{\includegraphics{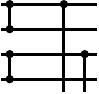}}}
    \hfill
    (ii)~\raisebox{-\height/2}{\makebox{\includegraphics{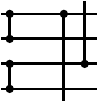}}}
    \hfill
    (iii)~\raisebox{-\height/2}{\makebox{\includegraphics{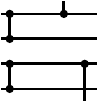}}}
    \hfill
    (iv)~\raisebox{-\height/2}{\makebox{\includegraphics{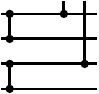}}}
    \hspace*\fill

    \caption{Possible cases for channels~$a$ and~$c$ in the proof of
      Theorem~\ref{thm:sat-thm}.  To obtain $C'$, add a comparator between channels~$b$ and $d$.}
    \label{fig:sat-thm}
  \end{figure}

  \begin{itemize}
  \item Figure~\ref{fig:sat-thm}~(i): $a$ and $c$ are min-channels at layer~$2$.

    Consider the input string $1^k11001^m$.  This is transformed to
    $1^k10011^m$ by $C'$, so $1^k10011^m\in\outputs(C')$.  We now show that
    $1^k10011^m\not\in\outputs(C)$.  In order to obtain the $0$ on channel~$b$, the
    input string would necessarily have a~$0$ on channel~$a$ because of the
    comparator~$(a,b)$ at layer~$1$.  But then the output would also have
    a~$0$ on channel~$a$, hence it could not be $1^k10011^m$.

  \item Figure~\ref{fig:sat-thm}~(ii): $a$ is a min-channel at layer~$2$, and $c$ is a max-channel.

    The argument is similar, but using the
    input string $0^k11001^m$.  This is transformed to $0^k10011^m$ by $C'$, so
    $0^k10011^m\in\outputs(C')$, and the same reasoning as above shows that
    $0^k10011^m\not\in\outputs(C)$.

  \item Figure~\ref{fig:sat-thm}~(iii): $a$ is a max-channel at layer~$2$, and $c$ is a min-channel.

    Consider again the input string $1^k11001^m$.  As before, this is
    transformed to $1^k10011^m$ by $C'$, so $1^k10011^m\in\outputs(C')$, and we
    show that $1^k10011^m\not\in\outputs(C)$.  As before, to obtain the $0$ on
    channel~$b$ the input string would necessarily have a~$0$ on channel~$a$ because
    of the comparator~$(a,b)$ at layer~$1$.  Now this $0$ is propagated upwards
    by the second-layer comparator at~$a$, which means that the output has
    a~$0$ on one of the first $k$ channels, hence it cannot be $1^k10011^m$.

  \item Figure~\ref{fig:sat-thm}~(iv): $a$ and~$c$ are both max-channels at layer~$2$.

    The reasoning is a bit more involved.
    Consider once more the input string $1^k11001^m$.  Since channel~$c$ is a
    second-layer max-channel connected w.l.o.g.~to a channel $j \leq k$, the output produced by~$C'$ is
    $1^{j-1}01^{k-j}10111^m$. In
    order to obtain this output with network~$C$, as before it is necessary to
    have inputs $0$ on channels~$a$ and~$b$; but since there are only two $0$s in
    the output, this means that channel~$a$ must also be connected to channel $j$
    on layer~$2$, which is impossible.
  \end{itemize}

\noindent  The cases where $a$ and $c$ are the unused (min) channels are similar.
\end{proof}

We believe the following generalization to hold.

\begin{conjecture}\label{conjecture}
  If the two-layer networks $C_1$ and $C_2$ on $n$ channels are both saturated and
  non-equivalent, then $\outputs(C_1)\not\subseteq\outputs(C_2)$.
\end{conjecture}

Particular cases of Conjecture~\ref{conjecture} are implied by
Theorem~\ref{thm:sat-thm}, but the general case remains open. The
conjecture has been verified experimentally for $n\leq 15$.

\section{Case studies: $n=4$ and $n=5$}
\label{sec:example}

This section provides a detailed analysis for the cases of
four-channel two-layer networks with first layer $F_4$ and
five-channel two-layer networks with first layer $F_5$.
Consider the following strategy to enumerate all possible
second layers: channel $1$ may be connected to channels~$2$, $3$
or~$4$, or may be unused; if channel~$1$ is connected to channel~$2$,
then channel~$3$ may be connected to channel~$4$ or may be unused;
etc.  With this strategy, the ten networks in Figure~\ref{fig:4wire} are generated in
the order $\mathit{abidjefgch}$.

\begin{figure}
\smallskip\emph{Redundant nets:}\smallskip

\fbox{%
  $a)$~\raisebox{-\height/2}{\includegraphics{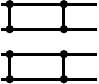}}
}
\fbox{%
  $b)$~\raisebox{-\height/2}{\includegraphics{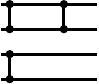}}
  $c)$~\raisebox{-\height/2}{\includegraphics{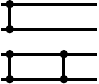}}
}

\smallskip\emph{Non-saturated nets:}\smallskip

\fbox{%
  $d)$~\raisebox{-\height/2}{\includegraphics{4-13.pdf}}
}
\fbox{%
  $e)$~\raisebox{-\height/2}{\includegraphics{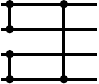}}
  $f)$~\raisebox{-\height/2}{\includegraphics{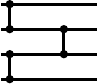}}
}

\smallskip
\fbox{%
  $g)$~\raisebox{-\height/2}{\includegraphics{4-24.pdf}}
}
\fbox{%
  $h)$~\raisebox{-\height/2}{\includegraphics{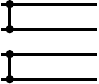}}
}

\smallskip\emph{Saturated nets:}\smallskip

\fbox{%
  $i)$~\raisebox{-\height/2}{\includegraphics{4-1324.pdf}}
}
\fbox{%
  $j)$~\raisebox{-\height/2}{\includegraphics{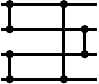}}
}\bigskip

\caption{The $10$ two-layer standard networks on four channels with the Parberry first layer $F_4$.}
\label{fig:4wire}
\end{figure}

The boxes around the networks represent classes of
equivalent networks.  There are only two non-trivial equivalence
classes.
The equivalence between nets $b)$ and $c)$ follows since
the permutation $(1\,3)(2\,4)$ transforms them into one another.  For
nets $e)$ and $f)$, applying the same permutation to $e)$ yields a net
that has a generalized comparator in layer~$2$; untangling it results in
$f)$.

The nets in the first row are all redundant, as they repeat a
comparator from the first layer; since the redundant comparators can
be removed without altering the set of outputs, they can be simplified to
net~$h)$.
The nets in the second row are not saturated; by
Theorem~\ref{thm:sat-char}, nets~$d)$ and~$g)$ are not saturated, and
their extension~$i)$ produces a subset of their outputs;
a similar situation arises with net~$e)$ vs net~$j)$, and net~$h)$
produces a superset of the outputs of both $i)$ and $j)$.  We detail the sets of
binary outputs for nets~$d)$, $g)$ and~$i)$, which correspond to
Case~$(3a)$ of Theorem~\ref{thm:sat-char}, one of the two cases missing
from the corresponding result in~\cite{DBLP:conf/lata/BundalaZ14}.
\[\begin{array}{rlcccccccccc}
\outputs(d) &= \{&0000,&0001,&0011,&0100,\\&&0101,&0110,&0111,&1111\phantom{,}&\} \\
\outputs(g) &= \{&0000,&0001,&0011,&0101,\\&&0111,&1001,&1101,&1111\phantom{,}&\} \\
\outputs(i) &= \{&0000,&0001,&0011,&0101,\\&&0111,&1111\phantom{,}&\}
\end{array}\]

For $n=5$ the situation is similar to $n=4$.  The generation algorithm from Section~\ref{sec:pathrep} produces the
26~two-layer networks in Figure~\ref{fig:5wire} in the order
$abdfqrshvtwixyckujzmenglop$.

\begin{figure}

\smallskip\emph{Redundant nets:}\smallskip

\fbox{%
  $a)$~\raisebox{-\height/2}{\includegraphics{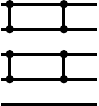}}
}
\fbox{%
  $b)$~\raisebox{-\height/2}{\includegraphics{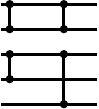}}
  $c)$~\raisebox{-\height/2}{\includegraphics{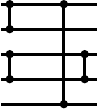}}
}
\fbox{%
  $d)$~\raisebox{-\height/2}{\includegraphics{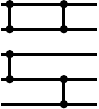}}
  $e)$~\raisebox{-\height/2}{\includegraphics{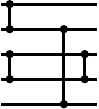}}
}

\smallskip
\fbox{%
  $f)$~\raisebox{-\height/2}{\includegraphics{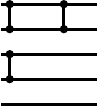}}
  $g)$~\raisebox{-\height/2}{\includegraphics{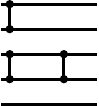}}
}

\smallskip\emph{Non-saturated nets:}\smallskip

\fbox{%
  $h)$~\raisebox{-\height/2}{\includegraphics{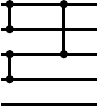}}
}
\fbox{%
  $i)$~\raisebox{-\height/2}{\includegraphics{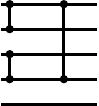}}
  $j)$~\raisebox{-\height/2}{\includegraphics{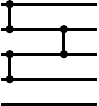}}
}
\fbox{%
  $k)$~\raisebox{-\height/2}{\includegraphics{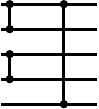}}
  $l)$~\raisebox{-\height/2}{\includegraphics{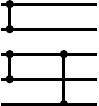}}
}

\smallskip
\fbox{%
  $m)$~\raisebox{-\height/2}{\includegraphics{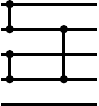}}
}
\fbox{%
  $n)$~\raisebox{-\height/2}{\includegraphics{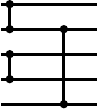}}
  $o)$~\raisebox{-\height/2}{\includegraphics{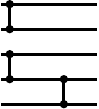}}
}
\fbox{%
  $p)$~\raisebox{-\height/2}{\includegraphics{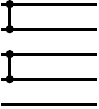}}
}

\smallskip\emph{Saturated nets:}\smallskip

\fbox{%
  $q)$~\raisebox{-\height/2}{\includegraphics{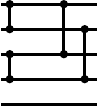}}
}
\fbox{%
  $r)$~\raisebox{-\height/2}{\includegraphics{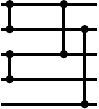}}
  $s)$~\raisebox{-\height/2}{\includegraphics{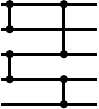}}
}
\fbox{%
  $t)$~\raisebox{-\height/2}{\includegraphics{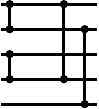}}
  $u)$~\raisebox{-\height/2}{\includegraphics{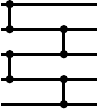}}
}

\smallskip
\fbox{%
  $v)$~\raisebox{-\height/2}{\includegraphics{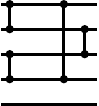}}
}
\fbox{%
  $w)$~\raisebox{-\height/2}{\includegraphics{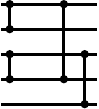}}
  $x)$~\raisebox{-\height/2}{\includegraphics{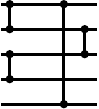}}
}
\fbox{%
  $y)$~\raisebox{-\height/2}{\includegraphics{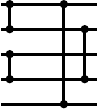}}
  $z)$~\raisebox{-\height/2}{\includegraphics{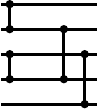}}
}

  \caption{The $26$ two-layer standard networks on five channels with the Parberry first layer $F_5$.}
  \label{fig:5wire}
\end{figure}

As before, the boxes identify the equivalence classes, which again can all be
obtained by means of the permutation $(1\,3)(2\,4)$ and eventually reversing any
generalized comparators at the second-layer.  The first set of networks is
redundant, while the second set is not saturated by Theorem~\ref{thm:sat-char},
and once again it can easily be verified that each network in this group
contains a set of outputs that is a proper superset of a network in the third
group.  Furthermore, only one element from each box in the third group needs to be
considered.

Following the notation in~\cite{DBLP:conf/lata/BundalaZ14}, we denote
the total number of two-layer networks on $n$ channels whose first
layer is $F_n$ by $|G_n|$; 
the number of non-equivalent such networks (up to permutation of
channels) by $|R(G_n)|$; 
and the corresponding values for saturated networks by $|S_n|$ and
$|R(S_n)|$.  
From these analyses, we obtain $|G_4|=10$, $|R(G_4)|=8$,
$|S_4|=|R(S_4)|=2$; and $|G_5|=26$, $|R(G_5)|=16$, $|S_5|=10$, and
$|R(S_5)|=6$.  The values for $|G_4|$, $|R(G_4)|$, $|G_5|$ and $|S_5|$
coincide with those in~\cite{DBLP:conf/lata/BundalaZ14}, whereas the
values we obtain for $|R(S_4)|$ and $|R(S_5)|$ coincide with those authors'
results after applying Lemma~\ref{lem:outputs} to eliminate representatives.
The difference in values in $|R(G_5)|$ and
$|R(S_5)|$ is probably due to an incomplete identification of the
equivalence classes (note that, for $n=5$, case~$(3)$ of
Theorem~\ref{thm:sat-char} is not necessary, so the notion of
saturated from \cite{DBLP:conf/lata/BundalaZ14} coincides with our
definition in the previous section).  The problem of computing the
equivalence classes \emph{efficiently} is the topic of the next
sections.

\section{Graph representation}

The results presented in \cite{DBLP:conf/lata/BundalaZ14} involve a
great deal of computational effort to identify permutations which
render various two-layer networks equivalent. Motivated by the
existence of sophisticated tools in the context of graph isomorphism,
we adopt a representation for comparator networks similar to the one
defined by Choi and Moon~\cite{DBLP:conf/gecco/ChoiM02}.
Let $C$ be a comparator network on $n$ channels.  The graph
representation of $C$ is a directed and labeled graph, $\GG(C)=(V,E)$
where each node in $V$ corresponds to a comparator in $C$ and
$E\subseteq V\times \{min,max\}\times V$.  Let $c(v)$ denote the
comparator corresponding to a node $v$. Then, $(u,\ell,v) \in E$ if
comparator $c(u)$ feeds into the comparator $c(v)$ in $C$ and the
label $\ell\in\{min,max\}$ indicates if the channel from $c(u)$ to
$c(v)$ is the min or the max output of $c(u)$.  Note that the
number of channels cannot be inferred from the graph representation,
as unused channels are not represented.

Each node has at most two in-edges and at most two out-edges. Nodes
with less than two in-edges represent comparators that are connected
to the input channels of the network. Similarly, nodes with less than two
out-edges represent comparators which are connected to the output
channels. As such, if the graph contains $k$ comparators, then the sum of
the in-degrees of the nodes and also the sum of the out-degrees of the
nodes is bounded by $2k-n$.

Clearly, graphs representing comparator networks are acyclic, and the
degrees of their vertices are bounded by~$4$.  There is a strong
relationship between equivalence of comparator networks and
isomorphism of their corresponding graphs. Choi and
Moon~\cite{DBLP:conf/gecco/ChoiM02} state the following proposition,
which implies that the comparator network equivalence problem is
polynomially reduced to the bounded-valence graph isomorphism problem.

\begin{proposition}\label{proposition:isomorphism}
  Let $C_1$ and $C_2$ be $n$-channel comparator networks. Then
\[ C_1\approx C_2 \Longleftrightarrow \GG(C_1) \approx \GG(C_2)\,.
\]
\end{proposition}

\begin{example}
  The sorting networks~$(a)$ and~$(b)$ from Page~\pageref{ex:sn} are
  represented by the following graphs,
  which can be seen to be isomorphic by mapping the vertices as
$a\mapsto v$, $b\mapsto u$, $c\mapsto w$, $d\mapsto x$, $e\mapsto y$ and
$f\mapsto z$.

{\small 
\[
\xymatrix@R-2em@C+1em{
  a \ar[r]^{min} \ar[rrdd]_(.2){max}
  & c \ar[r]^{min} \ar[rdd]^(.7){max}
  & d \ar[rd]^{max}
  \\
  &&& f
  \\
  b \ar[ruu]^(.3){max} \ar'[r]+0_{min}[rruu]
  && e \ar[ur]_{min}
}
\]
\[
\xymatrix@R-2em@C+1em{
  u \ar[rr]^{min} \ar[dr]_(.4){max}
  &&
  x \ar[rd]^{max}
  \\
  & w \ar[ru]_(.6){min} \ar[rd]^(.6){max}
  && z
  \\
  v \ar[ur]^(.4){min} \ar[rr]_{max}
  &&
  y \ar[ur]_{min}
}
\]
}
\end{example}

The graph isomorphism problem is one of a very small number of
problems belonging to NP, for which it is neither known that they are solvable in
polynomial time nor that they are NP-complete.  However, it is known that
the isomorphism of graphs of bounded valence (here: bounded degree) can be tested in
polynomial time~\cite{DBLP:journals/jcss/Luks82}, so the comparator
network equivalence problem can be efficiently solved.

An obvious approach for finding all two-layer prefixes modulo symmetry
is to generate all two-layer networks as demonstrated in
Section~\ref{sec:example}, and then apply graph isomorphism checking to find
canonical representatives of the equivalence classes. 
We evaluated this approach using the popular graph isomorphism tool
\verb!nauty!~\cite{DBLP:journals/jsc/McKayP14}, but found that the
exponential growth in the number of two-layer prefixes prevents this
approach from scaling.

Instead of a generate-and-test approach, in the next section we
present a scalable method for directly generating only one
representative two-layer prefix per equivalence class. Furthermore,
this approach also enables us to encode saturation as a syntactic
criterion in the generation process, i.e., to generate directly only
representatives of saturated two-layer prefixes.

\section{Path representation of two-layer networks}
\label{sec:pathrep}

In this section, we focus on two-layer networks where the first layer
is maximal (although not necessarily $F_n$).  These networks can be
uniquely represented in terms of the paths in their graph
representations. Furthermore, this representation can be read directly
from the network, and can be used to construct a canonical
representation of the network that completely characterizes the
equivalence classes in the generated graphs.  In the following, recall
that channels of a network are characterized as \emph{free},
\emph{min} or \emph{max} depending on the first layer.

\begin{definition}\label{def:path}
  A \emph{path} in a two-layer network $C$ is a sequence
  $\tuple{p_1p_2\ldots p_k}$ of distinct channels such that each pair
  of consecutive channels is connected by a comparator in $C$.
  The \emph{word} corresponding to $\tuple{p_1p_2\ldots p_k}$ is
  $\tuple{w_1w_2\ldots w_k}$, where $w_i$ is \sent{0}, \sent{1} or \sent{2}
  according to whether $p_i$ is the free channel, a min channel or a
  max channel, respectively.
\end{definition}

A path is maximal if it is a simple path (with no repeated nodes) that
cannot be extended (in either direction).  A network is connected if
its graph representation is connected.

\begin{definition}
  Let $C$ be a connected two-layer network on $n$ channels.  Then
  $\word(C)$ is defined as follows.

  \begin{description}
  \item[Head]%
    If $n$ is odd, then $\word(C)$ is the word
    corresponding to the maximal path in $C$ starting with
    the (unique) free channel. 
  \item[Stick]%
    If $n$ is even and $C$ has two channels not used
    in layer~$2$, then there are exactly two maximal paths in $C$ starting
    with a free channel (which are reverse to one another), and
    $\word(C)$ is the lexicographically smallest of
    the words corresponding these two paths.
  \item[Cycle]%
    If $n$ is even and all channels are used by a
    comparator in layer $2$, then $\word(C)$ is obtained by removing
    the last letter from the lexicographically
    smallest word corresponding to a maximal path in $C$ that
    begins with two channels connected in layer~$1$.
  \end{description}
\end{definition}

\begin{example}
  \label{ex:words}
  Below are three connected networks, ($a$), ($b$), and ($c$), with their
  maximal paths, pictured as ($a'$), ($b'$), and ($c'$), marked in bold.
  For instance, ($a'$) corresponds to the path $51243$.

  \noindent
  \hspace*\fill
  $(a)$~\raisebox{-\height/2}{\includegraphics{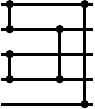}}
  \hfill
  $(b)$~\raisebox{-\height/2}{\includegraphics{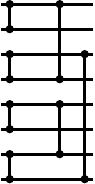}}
  \hfill
  $(c)$~\raisebox{-\height/2}{\includegraphics{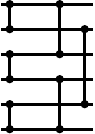}}
  \hspace*\fill\\[\smallskipamount]
  \hspace*\fill
  $(a')$~\raisebox{-\height/2}{\includegraphics{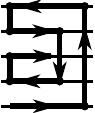}}
  \hfill
  $(b')$~\raisebox{-\height/2}{\includegraphics{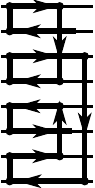}}
  \hfill
  $(c')$~\raisebox{-\height/2}{\includegraphics{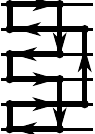}}
  \hspace*\fill\smallskip

  \noindent
  Network~$(a)$ involves an odd number of channels, and the word
  corresponding to the maximal path~$(a')$ starting on the free
  channel is $\sent{01221}$.
  Network~$(b)$ on an even number of channels contains two unused
  channels at layer~$2$, with two maximal paths~$(b')$ starting at a free
  channel and corresponding to the words $\sent{21212112}$ and
  $\sent{21121212}$ (its reverse); the corresponding word is thus the
  smallest of these two, namely $\sent{21121212}$. 
  Finally, network~$(c)$ consists of a cycle, $(c')$.  The words
  obtained by reading the possible maximal paths beginning with a
  layer~$1$ comparator are $\sent{122121}$ (starting on channel~$3$),
  $\sent{121221}$ and $\sent{122112}$ (starting on channels~$1$
  and~$5$, respectively, and proceeding in the reverse direction).
  The lexicographically smallest of these is $\sent{121221}$, and thus
  the corresponding word is $\sent{12122}$.
\end{example}

The set of all possible words (not necessarily minimal w.r.t.\ lexicographic ordering) can be described by the following BNF-style grammar.
\begin{align*}
  \mathsf{Word} &::= \mathsf{Head} \mid \mathsf{Stick} \mid \mathsf{Cycle}\\
  \mathsf{Head} &::= \sent{0}(\sent{12}+\sent{21})^\ast\\
   \mathsf{Stick} &::= (\sent{12}+\sent{21})^+ \\
   \mathsf{Cycle} &::= \sent{12}(\sent{12}+\sent{21})^\ast(\sent{1}+\sent{2})
\end{align*}

\begin{definition}
  The \emph{word representation} of a two-layer comparator network $C$, $\word(C)$, is the multi-set containing
  $\word(C')$ for each connected component $C'$ of $C$; we will denote this set
  by the ``sentence'' $w_1;w_2;\ldots;w_k$, where the words are in lexicographic
  order.
\end{definition}
In particular, a connected network will be represented by a sentence with only
one word, so there is no ambiguity in the notation $\word(C)$.  The
restriction that layer~$1$ be maximal corresponds to the requirement
that the multi-set $\word(C)$ have at most one Head-word.

\begin{wrapfigure}[9]{r}{0.2\textwidth}
  \vspace*{-3ex}
  \includegraphics{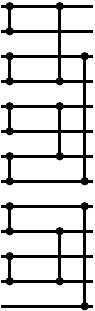}
  \quad
  \includegraphics{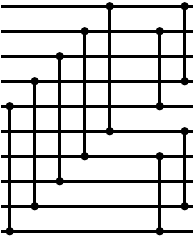}
\end{wrapfigure}
\ \vspace*{-\bigskipamount}

\begin{example}
  The first network on the right consists of two connected components,
  which are nets~$(a)$ and~$(b)$ of Example~\ref{ex:words}.
  It is therefore represented by the sentence containing the words
  corresponding to those nets, namely $\sent{01221;21121212}$.

  The second network consists of the first two layers of the $10$-channel
  sorting network from Figure~49 of~\cite{Knuth73}.  There are three
  connected components in this network, consisting of channels
  $\{1,4,6,9\}$, $\{2,5,7,10\}$ and $\{3,8\}$.  The first two components
  contain similar cycles represented by the word $\sent{122}$, while the
  third component yields the Stick-word $\sent{12}$.  The whole network
  is thus represented by the sentence $\sent{12;122;122}$.
\end{example}

Conversely, given a word $w=a_1\ldots a_\ell$ generated by the above grammar, we can
generate a corresponding two-layer network $\net(w)$ as follows.
\begin{enumerate}
\item The number of channels $n$ is:
  $|w|$, if $a_1=\sent{0}$ or $|w|$ is even;
  and $|w|+1$, if $|w|$ is odd and $a_1=\sent{1}$.
\item The first layer of $\net(w)$ is $F_n$.
\item If $w$ is a Stick-word or a Cycle-word, ignore the first character; then,
  for $k=0,\ldots,\left\lfloor\frac n2\right\rfloor-1$, take the next two
  characters $xy$ of $w$ and add a second-layer comparator between channels
  $2k+x$ and $2(k+1)+y$.  If $w$ is a Stick-word, ignore the last character;
  if $w$ is a Cycle-word, connect the two remaining channels at the end.
\item If $w$ is a Head-word, proceed as above but start by connecting the free
  channel to the channel indicated by the second character.
\end{enumerate}
This construction can be adapted straightforwardly to obtain a net with any
given first layer $L_1$: assuming the comparators in $L_1$ are numbered $1$ to
$\left\lfloor\frac n2\right\rfloor$, read $2k+x$ and $2(k+1)+y$ in step~$3$,
as ``the min/max channels from comparators $k$ and $k+1$'', where min or max
is chosen according to $x$ and $y$.

\begin{wrapfigure}[10]{h}{0.12\textwidth}
  \includegraphics{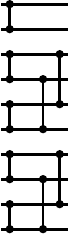}
  \quad\includegraphics{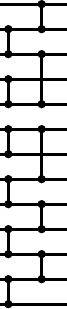}
\end{wrapfigure}

To generate a network from a sentence, simply generate the nets for each word in the sentence and compose them in the same order.

\begin{example}
  The two-layer networks on the right are generated from the sentences $\sent{12;122;122}$ and $\sent{01221;21121212}$, respectively.
  It can readily be seen that these networks are equivalent to the ones
  in the previous example.
\end{example}

\begin{lemma}
  Let $C$ and $C'$ be comparator networks on $n$ channels.  Then $C \approx C'$ iff $\word(C)=\word(C')$.
\end{lemma}
\begin{proof}
  The forward implication follows from the observation that, for
  two-layer networks, $C \approx C'$ means that there is a permutation
  $\pi$ such that $C'$ is $\pi(C)$ possibly with some generalized comparators
  in layer~$2$. Then, any path obtained in $C$ beginning at
  channel $j$ can be obtained in $C'$ by beginning at channel
  $\pi(j)$, and reciprocally.  The converse implication is
  straightforward.
\end{proof}

In algebraic terms, the function $\word$ can be seen as a
``forgetful'' functor that forgets the specific order of channels in a
net, whereas $\net$ generates the ``free'' net from a given word.
Furthermore, $\word$ always returns the minimum element in the ``fiber''
$\net^{-1}(w)$, whence lexicographically minimal words can be used to
characterize equivalent nets.  This means that $\word$ and $\net$ form
an adjunction between suitably defined pre-orders.

As a consequence, the sets of all distinct two-layer networks
on $n$ channels, $G_n$, and their equivalence classes modulo
permutations, $R(G_n)$, can be generated simply by generating all multi-sets of
words with at most one Head-word yielding exactly $n$ channels.  This
procedure has been implemented straightforwardly in Prolog, yielding
the values in the table of Figure~\ref{fig:table}. Besides the values
given in the table, $|R(G_{20})|=15{,}906$ was computed in a few
seconds, and $|R(G_{30})|=1{,}248{,}696$ in under a minute.

\begin{figure*}[t]
\[\begin{array}{c|r|r|r|r|r|r|r|r|r|r|r|r}
n & \multicolumn1{c|}{3}
 & \multicolumn1{c|}{4}
 & \multicolumn1{c|}{5}
 & \multicolumn1{c|}{6}
 & \multicolumn1{c|}{7}
 & \multicolumn1{c|}{8}
 & \multicolumn1{c|}{9}
 & \multicolumn1{c|}{10}
 & \multicolumn1{c|}{11}
 & \multicolumn1{c|}{12}
 & \multicolumn1{c|}{13}
 & \multicolumn1{c}{14}  \\ \hline
|G_n| & 4 & 10 & 26 & 76 & 232 & 764 & 2{,}620 & 9{,}496 & 35{,}696 & 140{,}152 & 568{,}504 & 2{,}390{,}480 \\
|S_n| & 2 & 4 & 10 & 28 & 70 & 230 & 676 & 2{,}456 & 7{,}916 & 31{,}374 & 109{,}856 & 467{,}716 \\
|R(G_n)| & 4 & 8 & 16 & 20 & 52 & 61 & 165 & 152 & 482 & 414 & 1{,}378 & 1{,}024 \\
|R(S_n)| & 2 & 2 & 6 & 6 & 14 & 15 & 37 & 27 & 88 & 70 & 212 & 136 \\ 
|R_n| & 1 & 2 & 4 & 5 & 8 & 12 & 22 & 21 & 48 & 50 & 117 & 94 \\ 
\hline
\end{array}\]
\[\begin{array}{c|r|r|r|r|r}
n
 & \multicolumn1{c|}{15}
 & \multicolumn1{c|}{16}
 & \multicolumn1{c|}{17}
 & \multicolumn1{c|}{18}
 & \multicolumn1{c}{19} \\ \hline
|G_n| & 10{,}349{,}536 & 46{,}206{,}736 & 211{,}799{,}312 & 997{,}313{,}824 & 4{,}809{,}701{,}440 \\
|S_n| & 1{,}759{,}422 & 7{,}968{,}204 & 31{,}922{,}840 & 152{,}664{,}200 & 646{,}888{,}154\\
|R(G_n)| & 3{,}780 & 2{,}627 & 10{,}187 & 6{,}422 & 26{,}796 \\
|R(S_n)| & 494 & 323 & 1{,}149 & 651 & 2{,}632 \\ 
|R_n| & 262 & 211 & 609 & 411 & 1{,}367 \\ 
\hline
\end{array}\]
\[\begin{array}{c|r|r|r|r|r|r|r|r|r|r|r|r}
n
 & \multicolumn1{c|}{20}
 & \multicolumn1{c|}{21}
 & \multicolumn1{c|}{22}
 & \multicolumn1{c|}{23}
 & \multicolumn1{c|}{24}
 & \multicolumn1{c|}{25}
 & \multicolumn1{c|}{26}
 & \multicolumn1{c|}{27}
 & \multicolumn1{c|}{28}
 & \multicolumn1{c|}{29}
 & \multicolumn1{c|}{30}
 & \multicolumn1{c}{31}
 \\ \hline
|R(S_n)| & 1{,}478 & 5{,}988 & 3{,}040 & 13{,}514 & 6{,}744 & 30{,}312 & 14{,}036 & 67{,}638 & 30{,}552 & 150{,}128 & 64{,}168 & 331{,}970 \\
|R_n| & 894 & 3{,}098 & 1{,}787 & 6{,}920 & 3{,}848 & 15{,}469 & 7{,}830 & 34{,}318 & 16{,}690 & 75{,}979 & 34{,}486 & 167{,}472 \\ \hline
\end{array}\]
\[\begin{array}{c|r|r|r|r|r|r|r|r|r}
n
 & \multicolumn1{c|}{32}
 & \multicolumn1{c|}{33}
 & \multicolumn1{c|}{34}
 & \multicolumn1{c|}{35}
 & \multicolumn1{c|}{36}
 & \multicolumn1{c|}{37}
 & \multicolumn1{c|}{38}
 & \multicolumn1{c|}{39}
 & \multicolumn1{c}{40} \\ \hline
|R(S_n)| & 138{,}122 & 731{,}000 & 291{,}090 & 1{,}604{,}790 & 622{,}136 & 3{,}511{,}250 & 1{,}313{,}262 & 7{,}663{,}112 & 2{,}792{,}966 \\
|R_n| & 73{,}191 & 368{,}143 & 152{,}503 & 806{,}710 & 322{,}891 & 1{,}763{,}133 & 676{,}431 & 3{,}843{,}848 & 1{,}429{,}836 \\ \hline
\end{array}\]
\caption{Table detailing the number of all distinct two-layer networks
on $n$ channels, $G_n$, the number of saturated such networks, $S_n$, the number of equivalence classes modulo permutations, $R(G_n)$, the number of saturated equivalence classes modulo permutations, $R(S_n)$, and the number of saturated equivalence classes modulo permutations and reflections, $R_n$. When searching for optimal-depth sorting networks, only networks extending $R_n$ need to be considered.}
\label{fig:table}
\end{figure*}

The sequence $|G_n|$ is actually known in Mathematics: it is sequence
A000085 in The On-Line Encyclopedia of Integer
Sequences,\footnote{\url{https://oeis.org/A000085}} and corresponds
(among others) to the number of self-inverse permutations on $n$
letters. The first two elements of the sequences coincide. Thus, to prove the above claim, it suffices to show that $|G_n|$ satisfies
the characteristic recurrence for that sequence.%

\begin{theorem}\label{thm:recurrence}
  $\left|G_n\right|=\left|G_{n-1}\right|+(n-1)\left|G_{n-2}\right|$ for $n\geq3$.
\end{theorem}

\begin{proof}
  The correspondence is at the level of layers, not of networks.  Consider the
  following two operations on layers.
  \begin{enumerate}
  \item Given a layer $L$, $L^\bullet$ is $L$ with an extra unused channel
    at the end.
  \item Given a layer $L$, $L^{k\bullet}$ is $L$ with two extra channels connected
    by a comparator: one between channels $k$ and $k+1$ of $L$, the other at the
    end.
  \end{enumerate}
  Given a layer $L'$ on $n$ channels, there is a unique way to write $L'$ as
  $L^\bullet$ or $L^{k\bullet}$ (according to whether the last channel of $L'$ is
  used), establishing the desired relationship.
  There seems to be no obvious relationship between the two layer networks
  containing $L'$ and $L$ as their second layers.
\end{proof}\bigskip

Alternatively to considering the recurrence, one could argue that $|G_n|$ corresponds to the number of matchings in a complete graph with $n$ nodes, since every comparator joins two channels.\footnote{Thanks to Daniel Bundala for pointing out this observation.}

The sequence $|R(G_n)|$ does not appear to be known already, and it does not have
such a simple description.  The following properties are however interesting.

\begin{theorem}\mbox{}\label{thm:interesting}
  \begin{enumerate}
  \item The number of non-equivalent redundant two-layer networks using $n$
    channels is $|R(G_{n-2})|$.
  \item For odd $n$, $\left|R(G_n)\right|=\left|R(G_{n-1})\right|+2\left|R(G_{n-2})\right|$.
  \end{enumerate}
\end{theorem}

\begin{proof}
  The proof is based on the word representation of the nets.
  \begin{enumerate}
  \item If $C$ is a redundant net, then the sentence $\word(C)$ contains $\sent{12}$.
    Removing one occurrence of this word yields a sentence corresponding to a
    network with $n-2$ channels.  This construction is reversible, so there are
    $|R(G_{n-2})|$ sentences corresponding to redundant networks on $n$ channels.

  \item If $n$ is odd, then $\word(C)$ contains exactly one word beginning with
    \sent{0}.  If this word is \sent{0}, then removing it yields a network with
    $n-1$ channels, and this construction is reversible.  Otherwise, removing the
    two last letters in this word yields a network with $n-2$ channels; since the
    removed letters can be \sent{12} or \sent{21}, this matches each network
    on $n-2$ channels to two networks on $n$ channels.
  \end{enumerate}
\end{proof}\bigskip

The construction we described does not take into account the notion of
saturation.  However,
the characterization of saturation given by Theorem~\ref{thm:sat-thm} is
straightforward to translate in terms of the word associated with a network.
\begin{corollary}
  Let $C$ be a two-layer network. Then $C$ is
  saturated if $w=\word(C)$ satisfies the following properties.
  \begin{enumerate}
  \item If $w$ contains \sent{0} or $\sent{12}$, then all other words
    in sentence $w$ are cycles.
  \item No stick in $w$ has length $4$.
  \item Every stick in $w$ begins and ends with the same symbol.
  \item If $w$ contains a head or stick ending with $c$, then every head or stick in $w$ ends with $c$, for $c\in\{\sent{1},\sent{2}\}$.
  \end{enumerate}
\end{corollary}

Thus, the set of saturated two-layer networks can be generated by using the
following restricted grammar.
\begin{align*}
  \mathsf{Word} &::= \mathsf{Head} \mid \mathsf{Stick} \mid \mathsf{Cycle}\\
  \mathsf{Stick} &::= \sent{12} \mid \mathsf{eStick} \mid \mathsf{oStick} \\
  \mathsf{Head} &::= \sent0 \mid \mathsf{eHead} \mid \mathsf{oHead}\\
  \mathsf{eStick} &::= \sent{12}(\sent{12}+\sent{21})^+\sent{21} \\
  \mathsf{eHead} &::= \sent0(\sent{12}+\sent{21})^\ast\sent{12}\\
  \mathsf{oStick} &::= \sent{21}(\sent{12}+\sent{21})^+\sent{12} \\
  \mathsf{oHead} &::= \sent0(\sent{12}+\sent{21})^\ast21\\
  \mathsf{Cycle} &::= \sent{12}(\sent{12}+\sent{21})^\ast(\sent{1}+\sent{2})
\end{align*}

Furthermore, sentences are multi-sets $M$ such that:
(i)~if $M$ contains the words \sent0\ or \sent{12}, then all other elements of
$M$ are cycles;
(ii)~if $M$ contains an $\mathsf{eHead}$ or $\mathsf{eStick}$, then it contains
no $\mathsf{oHead}$ or $\mathsf{oStick}$.
With these restrictions, generating all saturated networks for $n\leq 20$
can be done almost instantaneously.  The numbers $S_n$ of
saturated two-layer networks and $R(S_n)$ of equivalence classes
modulo permutation are given in the first four lines in the table of Figure~\ref{fig:table}.

Bundala and Z{\'a}vodn{\'y} mention that the number of two-layer
networks could further be restricted by considering
reflections~\cite{DBLP:conf/lata/BundalaZ14} (with acknowledgement to
D.E.~Knuth).  The reflection of a comparator network on $n$ channels is
the network obtained by replacing each comparator $(i,j)$ by the comparator
$(n-j+1,n-i+1)$; when the first layer is the set $F'_n$ of
comparators of the form $(i,n-i+1)$, reflection leaves it unchanged.
Furthermore, they show that a two-layer network with first layer $F'_n$ can
be extended to a sorting network if, and only if, its reflection can be
extended to a sorting network, hence reflections can be removed from
$R(S_n)$ when searching for optimal depth sorting networks. 

Since the word representation is defined for any first layer, this symmetry
break can be encoded by a similar technique as the one applied for saturation.
By removing $\mathsf{eStick}$ and $\mathsf{eHead}$ from the above grammar,
we directly generate only the $118$ representatives for $13$-channel networks described
in~\cite{DBLP:conf/lata/BundalaZ14}.  Furthermore, it is
possible to have distinct cycles whose reflections are equivalent
(but not equal); this brings the number
of relevant two-layer networks on $13$ channels to~$117$.  The last line in the table of
Figure~\ref{fig:table} above details the number $|R_n|$ of representatives modulo equivalence
and reflection for each value of $n\leq 40$.  We can compute the set
$R_{30}$ in less than one minute and $R_{40}$ in approximately two hours.

Having computed $R_{16}$, we can directly verify the known value $9$ for the optimal depth of a 16-channel sorting network, obtained only indirectly in \cite{DBLP:conf/lata/BundalaZ14}.
This direct proof involves
showing that none of the $211$
two-layer comparator networks in $R_{16}$ extends to a sorting network
of depth~$8$. For this, we use an encoding to Boolean
satisfiablity (SAT) as described in~\cite{DBLP:conf/lata/BundalaZ14},
where for each network $C$
in $R_{16}$, we generate a
formula $\varphi_C$
that
is satisfiable if and only if there exists a
sorting network of depth $8$
extending $C$.
Showing the unsatisfiability of these $211$ SAT instances can be performed in
parallel, with the hardest instance (a CNF with approx.\ $450{,}000$ clauses) requiring approx.\ $1800$ seconds running on a single thread of a cluster of Intel Xeon E5-2620 nodes clocked at $2$~GHz.

However, this approach does not directly work for $n=17$, where the best known
upper bound is $11$.
Attempting to show that there is no sorting network of depth $10$
requires analyzing the networks in
$R_{17}$. The resulting $609$ formulas have more than five
million clauses each, and none could be solved within a couple of weeks.
It appears that finding the optimal
depth of sorting networks with more than $16$ channels is a hard
challenge that will require prefixes with more than $2$ layers.

\section{Conclusion}

We presented an efficient technique to generate, modulo symmetry, the
set $R(S_n)$ of all saturated two-layer comparator networks on $n$~channels,
as well as its restriction $R_n$ to exclude networks
that are equivalent modulo reflection.

As noted by Parberry in~1991 and again by Bundala and Z{\'a}vodn{\'y} in~2014,
computing $R(S_n)$ and $R_n$ is a crucial step in the search for
optimal depth sorting networks on $n$~channels.
Using our approach we can compute $R(S_{13})$ in under a second vs
$30$~minutes using the brute force approach applied in \cite{DBLP:conf/lata/BundalaZ14},
and
improve the number of relevant two-layer prefixes to be considered
from $212$ to $117$ by eliminating
reflections.

In personal communication, Bundala and Z{\'a}vodny state that their
brute-force approach does not scale beyond $n=13$. This is not
suprising, as there is an exponential growth of the number of
candidate networks, a quadratic number of subsumption tests between
the candidate networks, and, for each subsumption test, a factorial number
of permutations and an exponential number of inputs to consider.

The smallest open instance of the optimal depth sorting network
problem is for $n=17$. We can easily compute $R(S_{17})$ (in $2$~seconds)
as well as its restriction to networks modulo
reflection. This later set consists of only $609$ networks and is a key
ingredient to solving this
problem,
effectively reducing the search space more than
$300{,}000$-fold.

\bibliographystyle{abbrv}
\bibliography{paper}

\end{document}